\documentclass[letter]{IEEEtran}
\usepackage[utf8]{inputenc}
\def\BibTeX{{\rm B\kern-.05em{\sc i\kern-.025em b}\kern-.08em
    T\kern-.1667em\lower.7ex\hbox{E}\kern-.125emX}}

\usepackage{graphicx}
\usepackage{color}
\usepackage{cite}
\usepackage{amsmath}
\usepackage{amssymb}
\usepackage{comment}
\usepackage{mathtools}
\usepackage{tabulary}
\usepackage{acronym}
\usepackage{upgreek}
\usepackage{algorithm}
\usepackage{algpseudocode}
\usepackage{placeins}
\usepackage[super]{nth}
\usepackage{dsfont}
\usepackage{xcolor}
\usepackage{flushend}
\usepackage{amsthm}
\usepackage{subcaption}
\usepackage[
    colorlinks=true,
    linkcolor=black,      
    citecolor=green,      
    urlcolor=blue,        
    filecolor=blue
]{hyperref}
\usepackage[acronym]{glossaries}
\makeglossaries

\newacronym{SLFN}{SLFN}{Single hidden Layer Feedforward Networks}
\newacronym{ELM}{ELM}{Extreme Learning Machine}
\newacronym{MLP}{MLP}{Multi-Layer Perceptron}
\newacronym{ML}{ML}{Machine Learning}
\newacronym{GO}{GO}{Goal-Oriented}
\newacronym{OTA}{OTA}{Over-The-Air}

\newacronym{MSE}{MSE}{Mean Squared Error}
\newacronym{LS}{LS}{Least Squares}
\newacronym{SVD}{SVD}{Singular Value Decomposition}
\newacronym{iid}{i.i.d.}{independent and identically distributed}

\newacronym{TX}{Tx}{Transmitter}
\newacronym{RX}{Rx}{Receiver}
\newacronym{AWGN}{AWGN}{Additive White Gaussian Noise}
\newacronym{LOS}{LoS}{Line-of-Sight}
\newacronym{NLOS}{NLoS}{Non-Line-of-Sight}
\newacronym{RF}{RF}{Radio Frequency}
\newacronym{MIMO}{MIMO}{Multiple Input Multiple Output}
\newacronym{ADC}{ADC}{Analog-to-Digital Converter}
\newacronym{SNR}{SNR}{Signal-to-Noise Ratio}
\newacronym{ULA}{ULA}{Uniform Linear Arrays}
\newacronym{SIM}{SIM}{Stacked Intelligent Metasurface}
\newacronym{RIS}{RIS}{Reconfigurable Intelligent Surface}
\newacronym{EM}{EM}{Electromagnetic}
\newacronym{XL}{XL}{eXtremely Large}

\newacronym{NPA}{NPA}{Nonlinear Power Amplifier}

\newacronym{D2D}{D2D}{Device-to-Device}
\newacronym{EI}{EI}{Edge Inference}
\newacronym{SL}{SL}{Supervised Learning}
\newacronym{SGD}{SGD}{Stochastic Gradient Descent}
\newacronym{MS}{MS}{Metasurface}
\newacronym{DNN}{DNN}{Deep Neural Network}
\newacronym{AR}{AR}{Autoregresive}

\newacronym{JSCC}{JSCC}{Joint Source-Channel Coding}
\newacronym{CMOS}{CMOS}{Complementary Metal-Oxide-Semiconductor}
\newacronym{SVM}{SVM}{Support Vector Machine}
\newacronym{TV}{TV}{Time Varying}
\newacronym{CSI}{CSI}{Channel State Information}

\newtheorem{condition}{Condition}

\newtheorem{proposition}{\bf Proposition}

\newtheorem{remark}{Remark}

\newcommand\iidsample{\stackrel{\mathclap{\tiny\mbox{i.i.d.}}}{\sim}}
\usepackage{etoolbox}
\usepackage{xcolor}

\title{Universal Approximation with XL MIMO Systems: OTA Classification via Trainable Analog Combining\\ \vspace{0.2cm}\normalfont\large \textit{(Extended version of a paper submitted to an IEEE Letters)}}

\author{Kyriakos Stylianopoulos,~\IEEEmembership{Student~Member,~IEEE}, and George C. Alexandropoulos,~\IEEEmembership{Senior~Member,~IEEE}
\thanks{This work has been supported by the SNS JU project 6G-DISAC under the EU's Horizon Europe research and innovation program under Grant Agreement No 101139130. The authors are with the Department of Informatics and Telecommunications,
National and Kapodistrian University of Athens, 16122 Athens, Greece (e-mails: \{kstylianop, alexandg\}@di.uoa.gr).}
}

\begin{document}

\maketitle

\begin{abstract}
In this paper, we show that an eXtremely Large (XL) Multiple-Input Multiple-Output (MIMO) wireless system with appropriate analog combining components exhibits the properties of a universal function approximator, similar to a feedforward neural network. By treating the channel coefficients as the random nodes of a hidden layer and the receiver's analog combiner as a trainable output layer, we cast the XL MIMO system to the Extreme Learning Machine (ELM) framework, leading to a novel formulation for Over-The-Air (OTA) edge inference without requiring traditional digital processing nor pre-processing at the transmitter. Through theoretical analysis and numerical evaluation, we showcase that XL-MIMO-ELM enables near-instantaneous training and efficient classification, even in varying fading conditions, suggesting the paradigm shift of beyond massive MIMO systems as OTA artificial neural networks alongside their profound communications role. Compared to conventional ELMs and deep learning approaches, whose training takes seconds to minutes, the proposed framework achieves on par performance (above $90\%$ classification accuracy across multiple data sets) with optimization latency of few milliseconds under the same number of trainable parameters, considering rich fading, low noise channels with XL receive antennas, making it highly attractive for inference tasks with ultra-low-power devices.

\end{abstract}

\section{Introduction}~\label{sec:introduction}
Future device-to-device and \gls{GO} networks will facilitate the communication of sensory data from \gls{TX} to \gls{RX} devices, not solely for the purpose of conventional information decoding and storage, but also for extracting features that guide autonomous devices towards desired actions~\cite{GO_review2, 6G-DISAC-magazine}. While either of the \gls{TX} and \gls{RX} may potentially solve a feature extraction task independently, this practice leads to large data rate needs when the \gls{TX} transmits the complete data for the \gls{RX} to perform the computational task, or large computational requirements at the (low power) \gls{TX}, if the computation is first carried out locally.

Alternatively, under \gls{EI}~\cite{DMB23}, the \gls{TX}-\gls{RX} system is treated as an end-to-end artificial neural network model trained on the particular data and channel conditions to perform arbitrary inference tasks. Learnable representations --typically  outputs of intermediate layers of \glspl{DNN}-- are thus exchanged \gls{OTA}. \gls{EI} can thus utilize resources more efficiently for the particular goal or task, while simultaneously endowing the system with a high level of flexibility, as either endpoint may be made almost arbitrarily lightweight computationally. \gls{EI} works based on semantic communications~\cite{DeepSC} have showcased advantages in tasks like image retrieval~\cite{JGM21_Image_Retrieval} and semantic alignment~\cite{paolo_semantic_alignment}.

Recently, ideas from \gls{OTA} computing~\cite{OTA_review} have been incorporated in \gls{EI} systems, with the intention of offloading part of the computations directly in the wave domain~\cite{OTA25_Rahimian}. The work in~\cite{DeepOAC} treated the wireless channel as a hidden, yet uncontrollable, \gls{DNN} layer. By positioning reflective or diffractive \glspl{MS}~\cite{XYN18_Diffractive_DNN, Momeni2022_wave_DL_computing} onto the wireless environment, the channel response may be controlled similar to typical \gls{DNN} layers offering notable improvements in communication and computation resources, as demonstrated by~\cite{Stylianopoulos_GO}. Instead, the authors of~\cite{GJZ24_SIM_TOC} placed stacked \glspl{MS} at the transceivers towards substituting layers by wave-domain computational devices. Another methodology has recently proposed in~\cite{Gunduz_Layer_Approximation}, where trained \gls{DNN} layers were approximated by \gls{MS}-parametrized \gls{MIMO} channels.

Nevertheless, such forms of \gls{OTA} approaches rely, at least partially, on digital \gls{DNN} computations at the endpoints, which limits the benefits of \gls{OTA} offloading. Additionally, the extent under which \gls{OTA} layers perform equivalent computations to \glspl{DNN} has been largely unexplored from a theoretical perspective. Key insights that spurred the widespread adoption of digital \glspl{DNN}, such as universal approximation and generalization, yet remain absent.
Training overheads when accounting for fading conditions might prove practical barriers for \gls{EI} deployment. In fact, works trained on static fading (e.g., \cite{paolo_semantic_alignment, DeepOAC, Gunduz_Layer_Approximation}) require complete retraining at every channel coherence frame, while works accounting for dynamic fading conditions (i.e.,~\cite{JGM21_Image_Retrieval, Stylianopoulos_GO, GJZ24_SIM_TOC}) converge slowly.
\color{black}

Motivated by the above, we first show that fundamental \gls{DNN} operations can be performed exclusively \gls{OTA}, in the wave propagation domain, without the need for digital processing at the transceivers. We consider an \gls{XL} \gls{MIMO} system, where the \gls{RX} is implemented by a large \gls{MS}-based antenna array of receiving elements with nonlinear response and analog combining capabilities~\cite{9324910}, and show that it can be trained according to the \gls{ELM} framework~\cite{Huang_ELM_2006}, which guarantees universal function approximation. The channel fading coefficients are treated as random hidden weights, while the analog combiner acts as the trainable output layer. Finally, the closed-form training of the proposed \gls{XL}-\gls{MIMO}-\gls{ELM} may fit within the channel coherence time, with the system being capable of adapting to channel changes through low overhead optimization.

\section{System Model}\label{sec:system-model}
Consider an \gls{XL} \gls{MIMO} system with $N_t$ \gls{TX} antenna elements and $N_r$  \gls{RX} metamaterial-based antennas. Let $\mathbf{s} \triangleq [s_1, \dots s_{N_t}]^\top\in \mathbb{C}^{N_t \times 1}$ be the transmitted signal and $\mathbf{\tilde{n}} \sim \mathcal{CN}(\mathbf{0}, \sigma^2 \mathbf{I}) \in \mathbb{C}^{N_r \times 1}$ be the \gls{AWGN}. According to Ricean fading \cite{Larsson_MaMIMO_Ricean} with Ricean factor $\kappa$, pathloss $P_L$, deterministic \gls{LOS} component between the \gls{TX} and the \gls{RX} $\mathbf{H}_{\text{LoS}}$, and \gls{NLOS} component $\mathbf{H}_{\text{NLoS}}$ capturing the multipath effects and is assumed to exhibit Rayleigh fading. The channel matrix $\mathbf{H}\in \mathbb{C}^{N_r \times N_t}$ is defined as follows:
\begin{equation}\label{eq:ricean}
    \mathbf{H} \triangleq \sqrt{\frac{\kappa}{1 + \kappa}} \sqrt{P_L} \mathbf{H}_{\text{LoS}} + \sqrt{\frac{1}{1 + \kappa}} \sqrt{P_L} \mathbf{H}_{\text{NLoS}}.
\end{equation} 
The baseband representation of the received signal $\mathbf{y} \triangleq [y_1, \dots, y_{N_r}]^\top \in \mathbb{C}^{N_r \times 1}$ at the \gls{RX} metasurface-based antennas can be mathematically expressed as:
\begin{equation}\label{eq:received signal}
    \mathbf{y} \triangleq \mathbf{H} \mathbf{s} + \mathbf{\tilde{n}}.
\end{equation}

Unless otherwise specified, we consider \gls{MIMO} channel realizations to be \gls{iid} samples, while the positions of the transceivers remain fixed.
When \gls{TV} channels are considered, we utilize a first-order \gls{AR} model for time step $k$:
\begin{equation}\label{eq:AR-channel}
    \mathbf{H}(k) = \eta \mathbf{H}(k-1) + (1-\eta) \boldsymbol{\Theta}(k),
\end{equation}
where $0 \! < \! \eta \leq \! 1$ is the constant \gls{AR} coefficient, capturing the channel's temporal correlation, and $\boldsymbol{\Theta}(k) \! \iidsample \! \mathcal{CN}( \mathbf{0}, \mathbf{I})$.

\subsection{Analog Combining with Nonlinear Soft Thresholding}\label{sec:analog-combining}
The choice of the analog combiner architecture is of particular importance for the proposed XL-MIMO-ELM system, since nonlinear operations are needed to ensure universal approximation.
To that end, we adopt \cite{active_RIS}'s \gls{MS}-based architecture as the \gls{RX} panel, where each antenna element is characterized by a nonlinear yet {\em fixed} response, followed by a {\em controllable} linear analog combiner weight. Hence, each $y_n$ from \eqref{eq:received signal} is independently passed through a basic component implementing soft thresholding. For example, using Rapp's model~\cite{Rapp_model}, the component's transfer function is expressed as:
\begin{equation}\label{eq:rapp-activation}
    g_{\rm rapp}(y_n) \triangleq y_n({1 + (y_n/y_{\rm sat})^{\alpha}})^{-1},
\end{equation}
where $y_{\rm sat}$ is a saturation signal threshold, while $\alpha$ controls the effect of the nonlinearity.
Note that this model can capture the transfer function properties of various \gls{RF} components, such as low-noise, variable-gain, or automatic gain control amplifiers.
The exact details of \gls{MS} elements designed to offer this response are left purposely unspecified, as designing metamaterial and \gls{RF} components that best address the needs of \gls{XL}-\gls{MIMO}-\glspl{ELM} opens a new direction of research.
An illustration of different instances of $g_{\rm rapp}(y_n)$ for various values of $y_{\rm sat}$ and $\alpha$ is given in Fig.~\ref{fig:rapp-activation-examples}.
In the sequel, we choose $|y_{\rm sat}| = 1.5$ and $\alpha = 2$ to obtain an activation function in $[-1, 1]$ that mimics the shape of the sigmoid function (i.e., $g(x) = 1/(1+\exp(-x))$), and is thus convenient for binary classification problems, as considered in this paper.
The output of each $g_{\rm rapp}(y_n)$ is multiplied by a complex weight $w_n$ capable of joint phase-amplitude control. Then, all weighted received signals are guided to the analog combiner to obtain the output of the \gls{XL} \gls{MIMO} system as follows ($\mathbf{w} \triangleq [w_1, \dots, w_{N_r}]^\top$):
\begin{equation}\label{eq:mimo-system-output}
    z \triangleq \sum_{n=1}^{N_r} w_n g_{\rm rapp}(y_n) =\mathbf{w}^\top g_{\rm rapp}(\mathbf{y}).
\end{equation}

\begin{figure}[t]
    \centering
    \includegraphics[width=\linewidth]{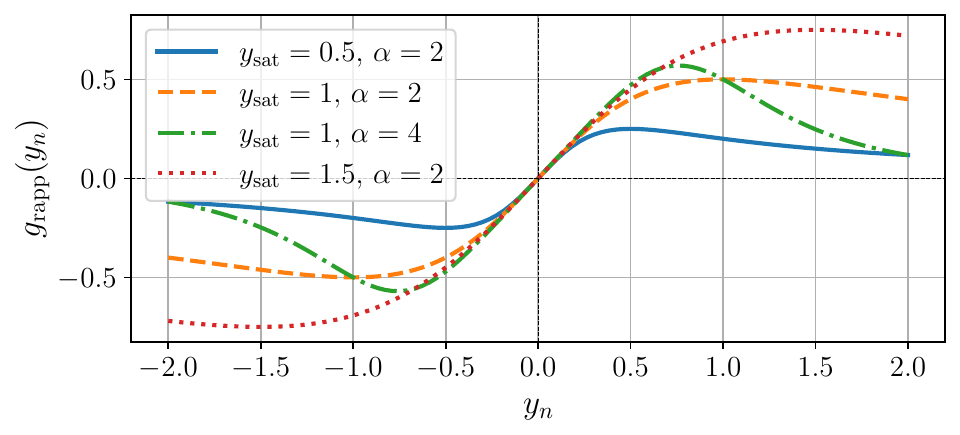}
    \caption{\small Soft thresholding response via Rapp's model~\cite{Rapp_model} which is used as the activation function for our XL-MIMO-ELM implemented directly with RF circuitry.}
    \label{fig:rapp-activation-examples}
\end{figure}



\begin{figure*}[t]
    \centering
    \includegraphics[scale=0.9,trim={14pt 10pt 0 0},clip]{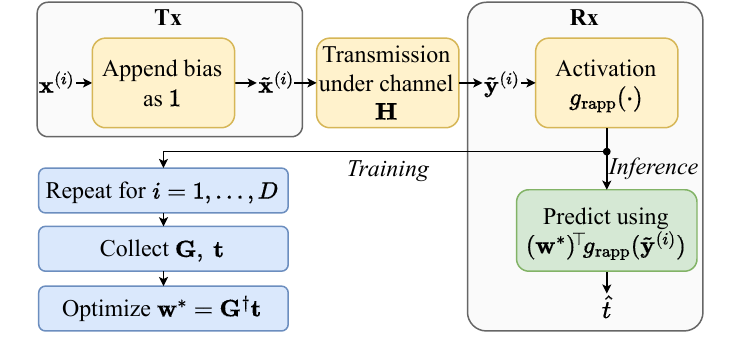}
    \caption{Flow of computation of the proposed \gls{XL}-\gls{MIMO}-\gls{ELM} methodology, including the training and inference procedures.}
    \label{fig:block_diagram}
\end{figure*}
\section{XL MIMO Learning Machines} \label{sec:model-description}
We consider a standard formulation of \gls{EI}, where a lightweight \gls{TX} observes correlated data instances through its sensors, and intends to transmit a computable feature of them to the \gls{RX}. A data set $\mathcal{D} \triangleq \{ (\mathbf{x}^{(i)}, t^{(i)})\}_{i=1}^{D}$ of $D$ observation vectors $\mathbf{x}^{(i)} \triangleq [x^{(i)}_1, \dots, x^{(i)}_d]^\top \in \mathbb{R}^{d \times 1}$, paired with their corresponding target values $t^{(i)} \in \mathbb{R}$, is available prior to the training process.
During the OTA supervised training phase, a mapping function between observations and targets $\hat{t} = f_{\mathbf{w}}(\mathbf{x})$ needs to be learned in one shot, that is parametrized by the analog combining vector $\mathbf{w}$.
The computations involved within $f_{\mathbf{w}}(\cdot)$ take place exclusively \gls{OTA} combinedly with analog processing at the \gls{RX}. 
In fact, the proposed model is expressed as a form of an \gls{SLFN}, where no processing takes place at the \gls{TX} other than what is required for analog transmission, as described in the sequel. Notwithstanding its parsimony, this choice has been made to investigate the theoretical characteristics of \gls{OTA} learning instead of proposing a more potent model. In particular, the \gls{TX} uses $N_t = d+1$ antennas to transmit the signal $\mathbf{\tilde{x}}^{(i)} \in \mathbb{R}^{d+1 \times 1}$, every element of which corresponds to the uncoded element of the observation vector $\mathbf{x}^{(i)}$, appended by $1$, i.e., $\mathbf{\tilde{x}}^{(i)} \triangleq [x^{(i)}_1, \dots, x^{(i)}_d,1]^\top$.
Despite work focusing on \glspl{ELM} on the complex domain~\cite{Complex_DNNs_Review}, we exclusively consider real-valued models, since they lend themselves more directly to theoretical investigations~\cite{Huang_ELM_2006, Huang_Activation} and are compatible with the utilized activation function. Therefore, only the real parts of the channel coefficients and \gls{AWGN}, which are denoted as $\mathbf{H}^{\rm r}$ and $\mathbf{\tilde{n}}_r$, are needed.
Leveraging~\eqref{eq:received signal} and considering the high \gls{SNR} regime where the \gls{AWGN} may be ignored, the \gls{OTA} transmission acts as \gls{SLFN}'s hidden linear layer, with the received signal be expressed as:
\begin{equation}\label{eq:hidden-layer-output}
    \mathbf{\tilde{y}}^{(i)} = \mathbf{H}^{\rm r}\mathbf{\tilde{x}}^{(i)} = \mathbf{H}^{\rm r}_{:,1:d}\mathbf{x}^{(i)} + \mathbf{h}^{\rm r}_{:,d+1},
\end{equation}
where $\mathbf{H}^{\rm r}_{:,1:d}$ and $\mathbf{h}^{\rm r}_{:,d+1}$ denote the first $d$ columns and the last ($d+1$) column of $\mathbf{H}^{\rm r}$, respectively, and the right hand side of~\eqref{eq:hidden-layer-output} is given to illustrate the implicit bias term of standard \glspl{SLFN} which arises due to the inclusion of the element $1$ on $\mathbf{\tilde{x}}^{(i)}$.
Next, each component of $\mathbf{\tilde{y}}^{(i)}$ is passed by the soft thresholding component described in Section~\ref{sec:analog-combining} which is used as a nonlinear activation function for the hidden layer.
Finally, the linear combiner $\mathbf{w}$ (now seen as an  $N_r$-dimensional real vector), acts as the trainable model weight:
\begin{align}\label{eq:mimo-elm-model}
    \hat{t}^{(i)} &= f_{\mathbf{w}}\left(\mathbf{x}^{(i)}\right) = \mathbf{w}^\top g_{\rm rapp}\left( \mathbf{\tilde{y}}^{(i)}\right) \\
     &= \mathbf{w}^\top g_{\rm rapp}\left( \mathbf{H}^{\rm r}\mathbf{\tilde{x}}^{(i)} \right) = \sum_{j=1}^{N_r} w_j g_{\rm rapp}\left( \mathbf{H}^{\rm r}_{j,:}\mathbf{\tilde{x}}^{(i)} \right).
\end{align}
The \gls{SLFN} of~\eqref{eq:mimo-elm-model}, where a hidden layer of \emph{uncontrollable, random} coefficients is followed by a nonlinear activation and a linear \emph{controllable} operation has the form of an \gls{ELM}~\cite{Huang_ELM_2006}, providing a framework for training and theoretical analysis.

\subsection{Training and Theoretical Guarantees} \label{sec:theoretical-analasyis}
Let $\mathbf{G} \in \mathbb{R}^{D \times N_r}$ denote the outputs of the hidden layer of $f_{\mathbf{w}}(\cdot)$ once passed through the activation function for all data points $\mathbf{x}^{(i)} \in \mathcal{D}$, i.e., $\mathbf{G} \triangleq [g_{\rm rapp}(\mathbf{H}^{\rm r}\mathbf{\tilde{x}}^{(1)}), \dots, g_{\rm rapp}(\mathbf{H}^{\rm r}\mathbf{\tilde{x}}^{(D)})]^\top$.
Similarly, let $\mathbf{t} \triangleq [t^{(1)}, \dots, t^{(D)}]^\top \in \mathbb{R}^{D \times 1}$ include the target values, so that the squared error between the model's outputs and target values over $\mathcal{D}$, seen as a cost function over $\mathbf{w}$, is expressed as:
\begin{equation}\label{eq:training}
    \mathbf{w}^\ast \triangleq {\rm arg}\!\min_{\mathbf{w}} \| \mathbf{G} \mathbf{w}   -  \mathbf{t}\|_{2}^{2} = \mathbf{G}^{\dagger}\mathbf{t}.
\end{equation}
The right hand side (r.h.s.) of~\eqref{eq:training} is the \gls{LS} solution~\cite[eq. (13) and Thm. 5.1]{Huang_ELM_2006}, where $\mathbf{G}^{\dagger}$ denotes the Moore-Penrose generalized inverse of $\mathbf{G}$. The overall methodology is illustrated in Fig.~\ref{fig:block_diagram}.

\begin{remark}\label{rem:weight-norm}
    The $\mathbf{w}^\ast$ obtained by~\eqref{eq:training} is unique while having the smallest norm among all possible solutions~\cite[Prop. 8.4.2]{Serre_Book}, and therefore inherently minimizes the system's reception power, which is given by $P_r = | \mathbf{w}|^2$.
\end{remark}


\begin{proposition}\label{prop:universal-approximation}
    Consider the \gls{ELM} expressed via~\eqref{eq:mimo-elm-model} with $\mathbf{H}^{\rm r}$ following the Ricean fading channel model with sufficiently small $\kappa$ (i.e., rich scattering) and the Rapp activation function of~\eqref{eq:rapp-activation} with $\alpha \in \mathbb{N}_+^* \backslash \{1\} $. Then, given any arbitrarily small value $\epsilon > 0$, there exists $N_r \leq D$ such that, for $D$ arbitrary distinct samples of $\mathcal{D} = \{ (\mathbf{x}^{(i)}, t^{(i)})\}_{i=1}^{D}$, there exists $\mathbf{w}^\ast$ so that $\| \mathbf{G} \mathbf{w}^\ast  \! - \! \mathbf{t} \|< \epsilon $ with probability $1$.
\end{proposition}
\begin{proof}
    The proof follows the direct application of~\cite[Thm~2.2]{Huang_ELM_2006} with a change of notation and the activation function requirements of~\cite{Huang_Activation}, under the conditions:
    \begin{condition}\label{cond:random-weights}
    The entries of $\mathbf{H}^r$ are \gls{iid} from a continuous probability distribution with infinite support over $\mathbb{R}^{N_r \times d+1}$.
    \end{condition}
    \begin{condition}\label{cond:activation}
    The function $g_{\rm rapp}(\cdot)$ is an infinitely differentiable nonlinear function. Further conditions on the activation functions are imposed by \cite{Huang_Activation}, stating $g_{\rm rapp}(\cdot)$ to be bounded and its limit as $x\to \infty$ or $x \to -\infty$ to exist.
    \end{condition}
    For $\kappa \to 0$, \eqref{eq:ricean}'s Ricean fading on $\mathbf{H}$ reduces to the uncorrelated Rayleigh distribution. Since it is continuous, $\mathbf{H}$, and hence, $\mathbf{H}^{\rm r}$ may be sampled in any interval in $\mathbb{R}^{N_r \times d+1}$ with nonzero probability, which fulfills Condition~\ref{cond:random-weights}. Also, \eqref{eq:rapp-activation}'s Rapp activation function is nonlinear and infinitely differentiable with respect to $y_n$, as long as $\alpha \in \mathbb{N}^+$. Furthermore, $\lim_{y_n \to \infty} g_{\rm rapp}(y_n) = 0$ for $\alpha > 1$, due to the dominance of the denominator.  Regarding boundedness for $y_n  \in (0, +\infty)$, $d g_{\text{rapp}}(y_n)/{dy_n} = y_{\text{sat}}^\alpha \left(y_{\text{sat}}^\alpha + (1 - \alpha) y_n^\alpha \right)/{(y_{\text{sat}}^\alpha + y_n^\alpha)^2} = 0$ provides a unique solution at $y_n^\ast = y_{\text{sat}} (\alpha - 1)^{-1/\alpha}$, and $g_{\rm rapp}(y_n^\ast) = \frac{y_{\text{sat}}}{\alpha} (\alpha - 1)^{1 - \frac{1}{\alpha}}$ is a finite maximum value. Following the same arguments for $y_n  \in (-\infty, 0)$ and by noting that $g_{\rm rapp}(0) = 0$, $g_{\rm rapp}(\cdot)$ is bounded everywhere. As a result, Condition~\ref{cond:activation} holds for the Rapp activation function, which completes the proof. 
\end{proof}
\begin{remark}
    The time complexity for obtaining the \gls{LS} solution of~\eqref{eq:training} is $\Theta(D N_r \min\{D, N_r\})$, resulting from the \gls{SVD} process in obtaining $\mathbf{G}^{\dagger}$. For reasonably small data sets systems where $D = \Theta(N_r)$, training may effectively take place within typical low mobility sub-$6$ GHz channel coherence frames of few ms~\cite{kivinen2001_emperical}.
\end{remark}

\begin{figure*}[t]
    \centering
    \begin{subfigure}[b]{0.45\textwidth}
        \centering
        \includegraphics[width=\linewidth]{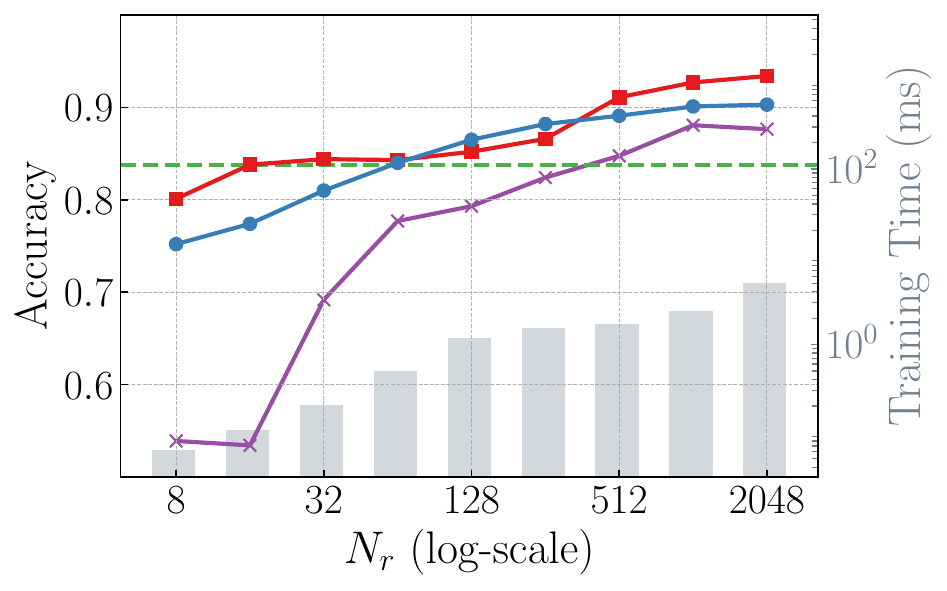}
        \caption{Parkinson's}
        \label{fig:parkinsons}
    \end{subfigure}~~~~%
    \begin{subfigure}[b]{0.45\textwidth}
        \centering
        \includegraphics[width=\linewidth]{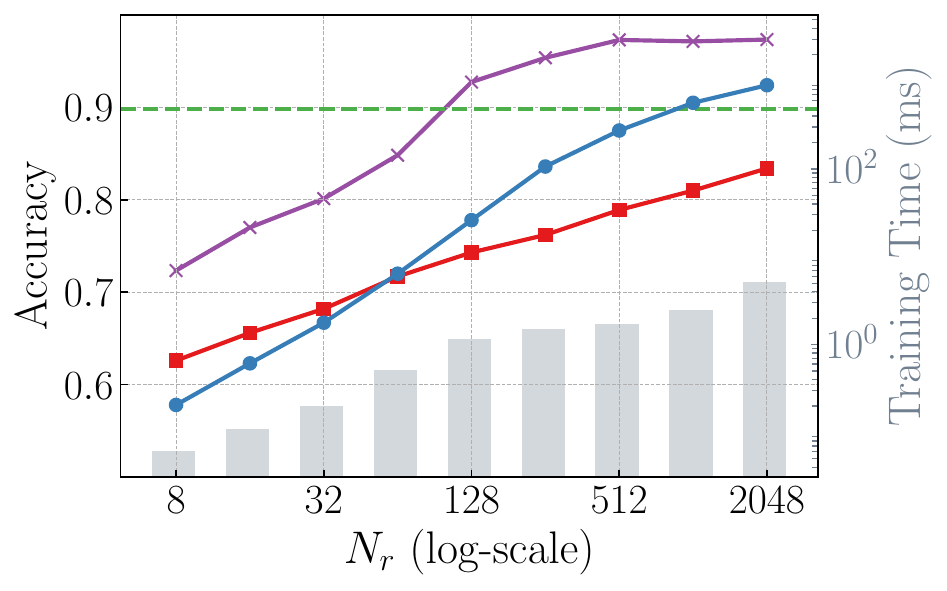}
        \caption{WBCD}
        \label{fig:breast_cancer}
    \end{subfigure}

    \begin{subfigure}[b]{0.45\textwidth}
        \centering
        \includegraphics[width=\linewidth]{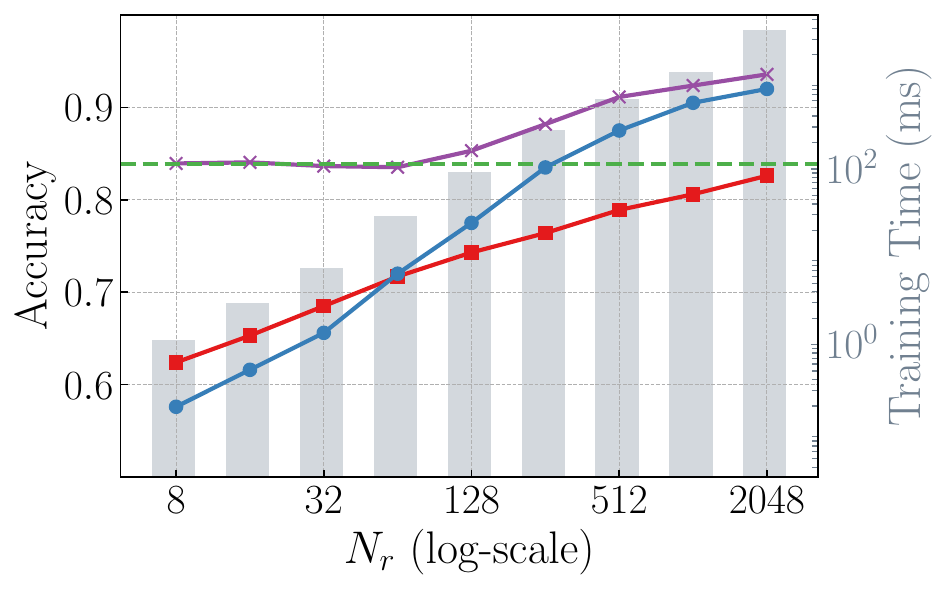}
        \caption{MNIST}
        \label{fig:mnist50}
    \end{subfigure}~~~~%
    \begin{subfigure}[b]{0.45\textwidth}
        \centering
        \includegraphics[width=\linewidth]{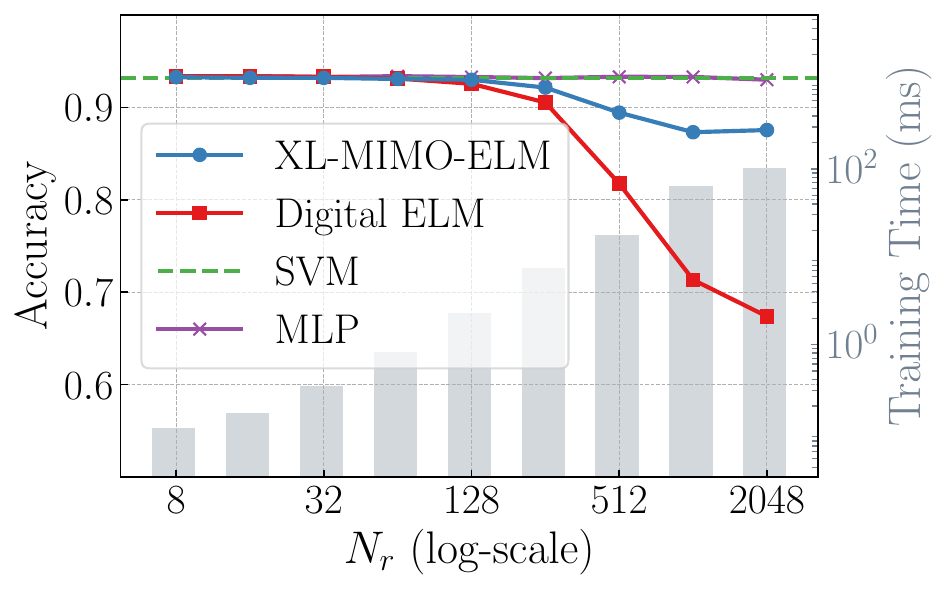}
        \caption{SECOM}
        \label{fig:secom}
    \end{subfigure}

    \caption{\small Comparative performance of the proposed \gls{XL}-\gls{MIMO}-\gls{ELM} with benchmarks across different datasets for increasing number of \gls{RX} antennas $N_r$ (equaling the number of units in the hidden layer). Bar graphs corresponding to the right vertical axis (log-scale) indicate the mean execution time for optimizing the \gls{XL}-\gls{MIMO}-\gls{ELM}.}
    \label{fig:all_datasets_wide}
\end{figure*}
The application of~\eqref{eq:mimo-elm-model} implies that both training and inference on unseen data must take place insofar $\mathbf{H}$ remains static, which is constraining in real-life communication systems. Assuming the \gls{AR} \gls{TV} model for the channel following~\eqref{eq:AR-channel} from time instance $k-1$ to $k$, we further propose a lightweight re-training policy to obtain approximately optimal weight vectors $\mathbf{w}^*(k)$ for every new channel realization $\mathbf{H}(k)$ sampled through~\eqref{eq:AR-channel}.
A set of random data point indices $\mathcal{S}$ is sampled to obtain a small mini-batch.
Its corresponding hidden-layer output matrix and target vector are computed as $\mathbf{G}_{\mathcal{S}} \triangleq [g_{\rm rapp}(\mathbf{H}^{\rm r}(k)\mathbf{\tilde{x}}^{(i)})]_{i \in \mathcal{S}}^\top$ and $\mathbf{t}_{\mathcal{S}} \triangleq [t^{(i)}]_{i \in \mathcal{S}}^\top$.
Similar to \eqref{eq:training}, $\mathbf{w}_{\mathcal{S}}^* \triangleq \mathbf{G}_{\mathcal{S}}^\dagger \mathbf{t}_{\mathcal{S}}$ is obtained for the mini-batch \gls{LS} problem with complexity $\Theta(|\mathcal{S}|^2N_r)$ and is then used to update $\mathbf{w}(k)$:
\begin{equation}\label{eq:w-online-update}
    \mathbf{w}^*(k) \gets \mathbf{w}^*(k) + \gamma \mathbf{w}_{\mathcal{S}}^*,
\end{equation}
where $\gamma < 1$ is the learning rate.
The sampling, \gls{LS}-solution, and update steps are repeated until convergence, while $\mathbf{w}^*(k)$ is initialized as the optimal solution for the $k-1$-th time step.
In practice, very few iterations are needed since $\mathbf{H}(k)$ and $\mathbf{H}(k-1)$ have high cross-correlation for typical wireless environments, leading to good starting points for $\mathbf{w}^*(k)$.

\begin{figure*}[!t]
    \centering
    \begin{subfigure}[b]{0.4\textwidth}
        \centering
        \includegraphics[width=\linewidth]{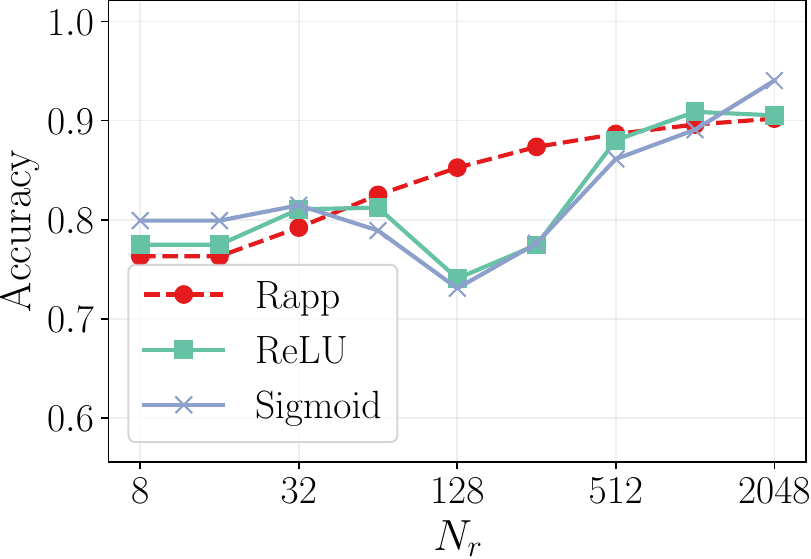}
        \caption{Parkinson's}
    \end{subfigure}~~%
    \begin{subfigure}[b]{0.4\textwidth}
        \centering
        \includegraphics[width=\linewidth]{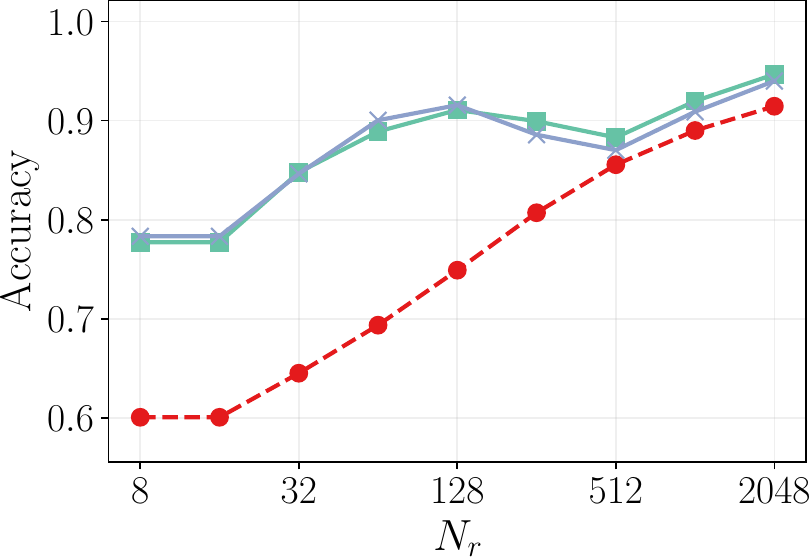}
        \caption{WBCD}
    \end{subfigure}

    \begin{subfigure}[b]{0.4\textwidth}
        \centering
        \includegraphics[width=\linewidth]{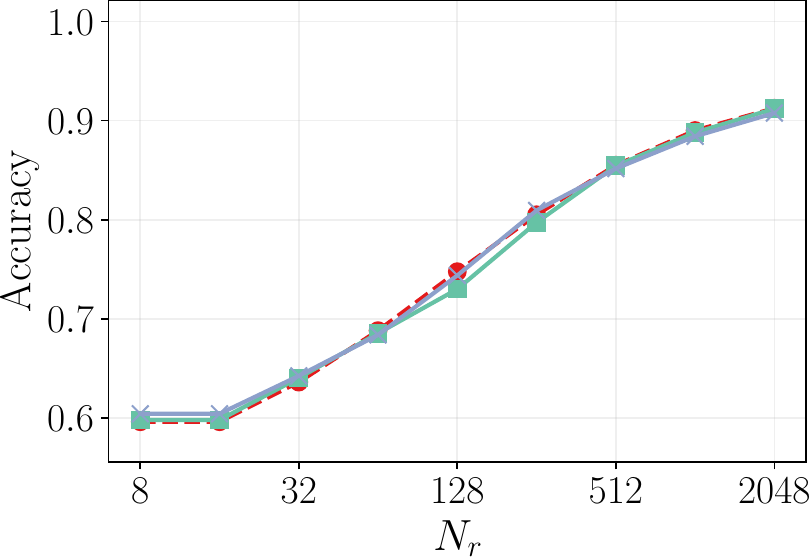}
        \caption{MNIST}
    \end{subfigure}~~%
    \begin{subfigure}[b]{0.4\textwidth}
        \centering
        \includegraphics[width=\linewidth]{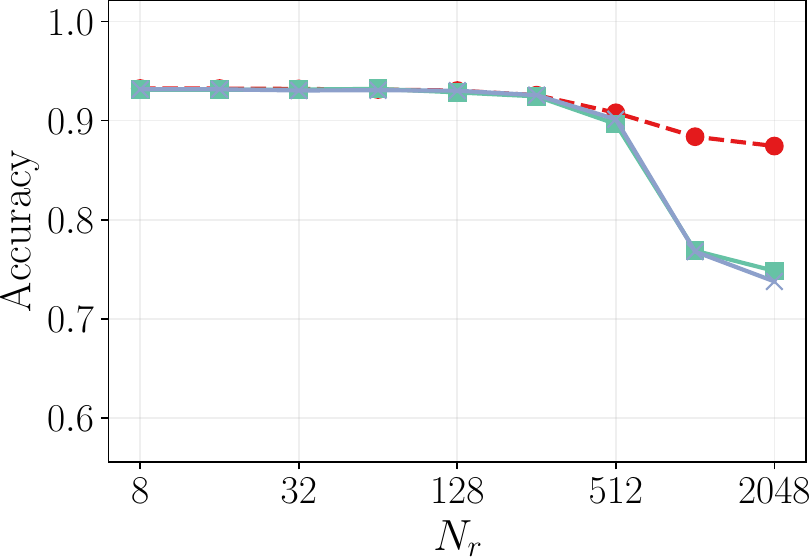}
        \caption{SECOM}
    \end{subfigure}

    \caption{Comparison of different activations under the \gls{XL}-\gls{MIMO}-\gls{ELM} methodology. The proposed ``Rapp'' activation function is implemented using analog \gls{RF} circuitry, while Sigmoid and ReLU are standard digital activations.}
    \label{fig::ablation-study}
\end{figure*}

\begin{figure*}[t]
    \centering
    \begin{subfigure}[b]{0.4\textwidth}
        \centering
        \includegraphics[width=\linewidth]{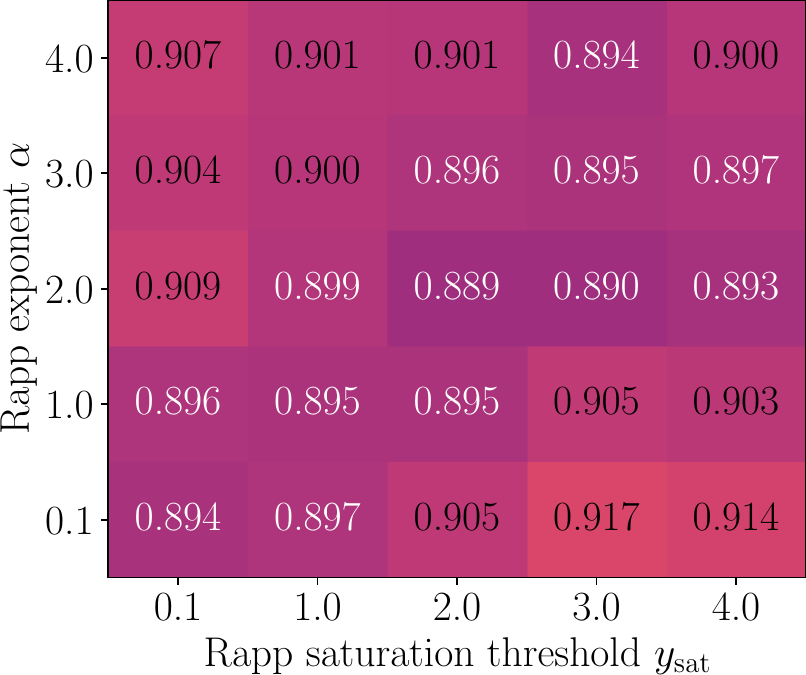}
        \vspace{0.002cm}
        \caption{Parkinson's}
    \end{subfigure}~~~~%
    \begin{subfigure}[b]{0.4\textwidth}
        \centering
        \includegraphics[width=\linewidth]{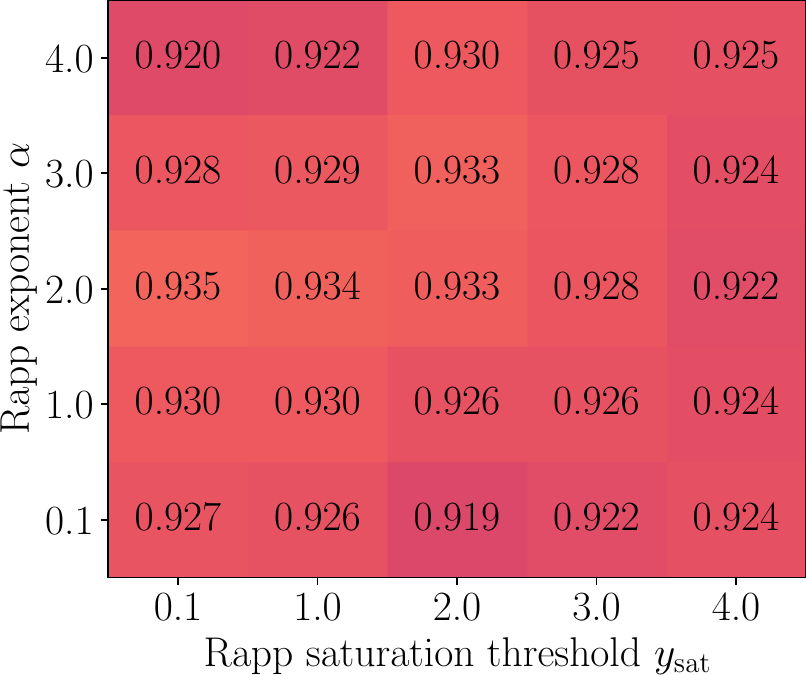}
        \vspace{0.002cm}
        \caption{WBCD}
    \end{subfigure}

    \vspace{0.8cm}
    
    \begin{subfigure}[b]{0.4\textwidth}
        \centering
        \includegraphics[width=\linewidth]{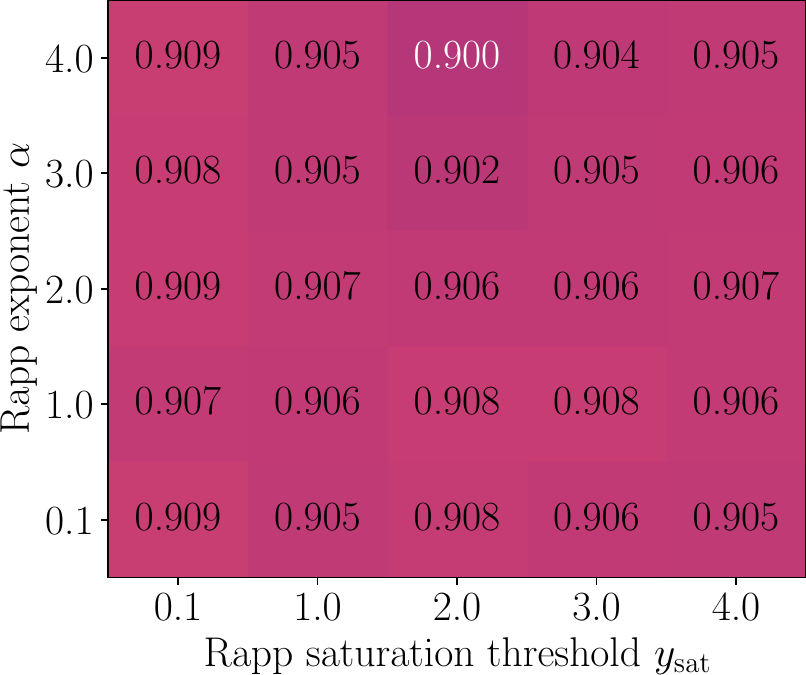}
        \vspace{0.002cm}
        \caption{MNIST}
    \end{subfigure}~~~~%
    \begin{subfigure}[b]{0.4\textwidth}
        \centering
        \includegraphics[width=\linewidth]{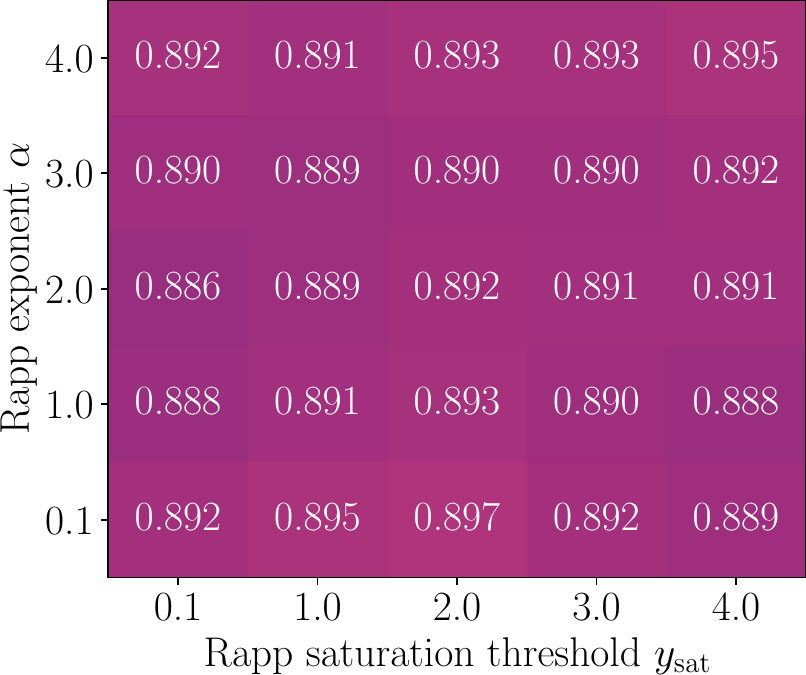}
        \vspace{0.002cm}
        \caption{SECOM}
    \end{subfigure}

    \caption{Sensitivity study of mean accuracy over various data sets for the parameters $\alpha$ and $y_{\rm sat}$ of the utilized Rapp activation function in~\eqref{eq:rapp-activation}. $N_r=1024$ has been used in all data sets except for SECOM, where $N_r$ was set to $512$ to avoid overfitting.}
    \label{fig:rapp-parameters}
\end{figure*}

\section{Numerical Evaluation}\label{sec:results}

\begin{figure*}[t]
    \centering
    \begin{subfigure}[b]{0.4\textwidth}
        \centering
        \includegraphics[width=\linewidth]{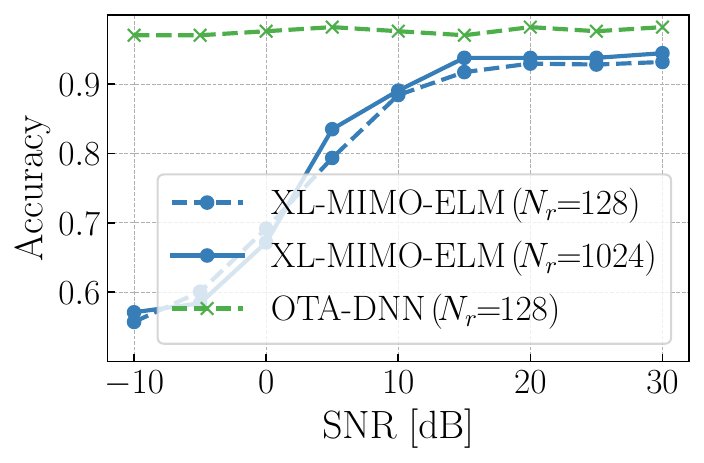}
        \caption{WBCD}
        \label{fig:snr_breast_cancer}
    \end{subfigure}~~~~%
    \begin{subfigure}[b]{0.4\textwidth}
        \centering
        \includegraphics[width=\linewidth]{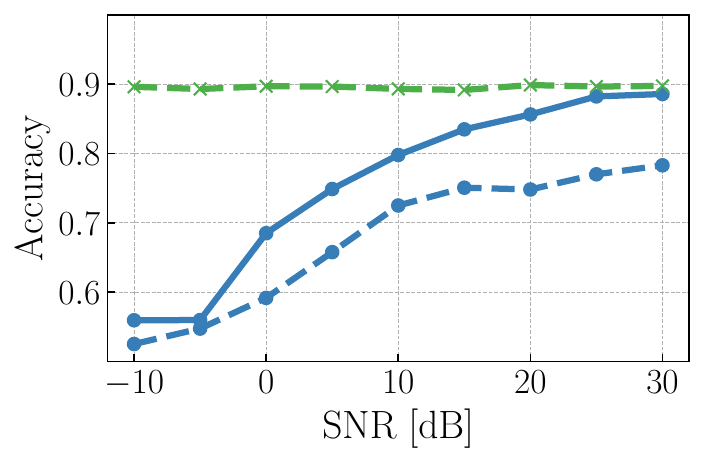}
        \caption{MNIST}
        \label{fig:snr_mnist}
    \end{subfigure}
        \caption{Performance of XL-MIMO-ELM accounting for the effects of \gls{AWGN} over different data sets and numbers of antennas (trainable parameters) compared to noise-free cases.}
        \label{fig:acc-vs-snr}
\end{figure*}


\begin{figure*}[t]
    \centering
    \begin{subfigure}[b]{0.4\textwidth}
        \centering
        \includegraphics[width=\linewidth]{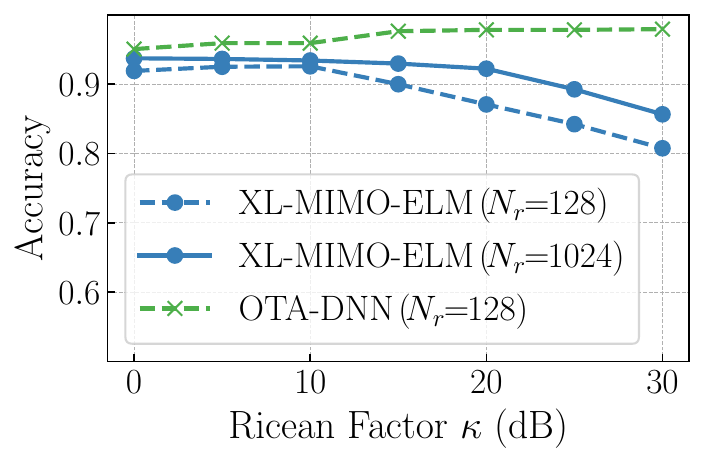}
        \caption{Ricean (WBCD)}
        \label{fig:ricean_breast_cancer}
    \end{subfigure}~~~~%
    \begin{subfigure}[b]{0.4\textwidth}
        \centering
        \includegraphics[width=\linewidth]{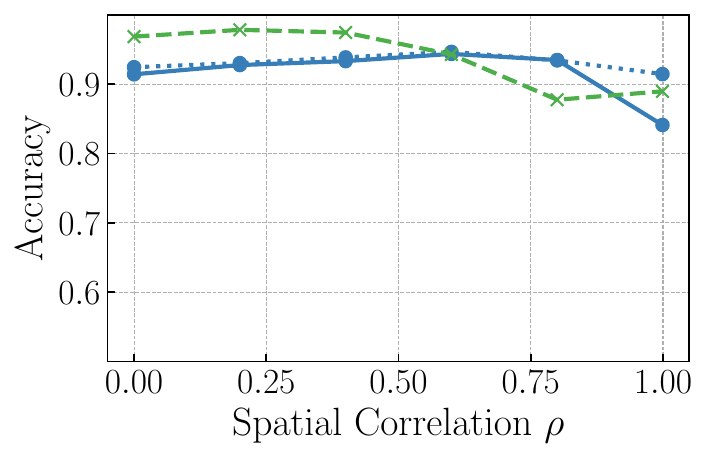}
        \caption{Correlated Rayleigh (WBCD)}
        \label{fig:spatial_correl}
    \end{subfigure}
        \caption{Performance of XL-MIMO-ELM over channel diversity levels on different data sets and numbers of antennas.}
        \label{fig:acc-vs-ricean}
\end{figure*}

\begin{figure*}[!t]
    \centering
    \begin{subfigure}[b]{0.4\textwidth}
        \centering
        \includegraphics[width=\linewidth]{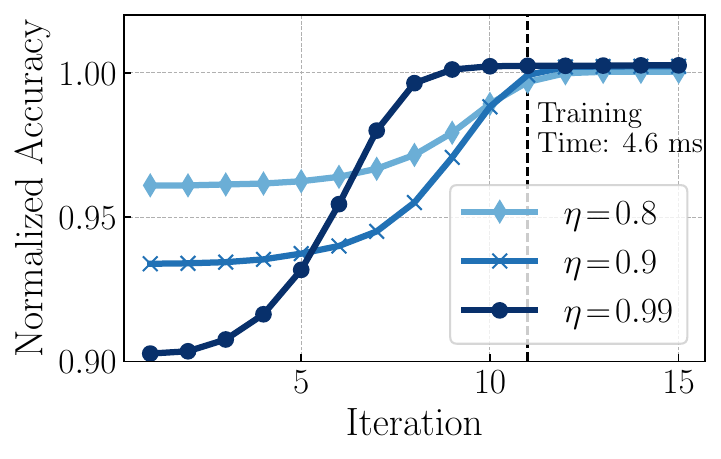}
        \caption{WBCD}
        \label{fig:online_breast_cancer}
    \end{subfigure}~~~~%
    \begin{subfigure}[b]{0.4\textwidth}
        \centering
        \includegraphics[width=\linewidth]{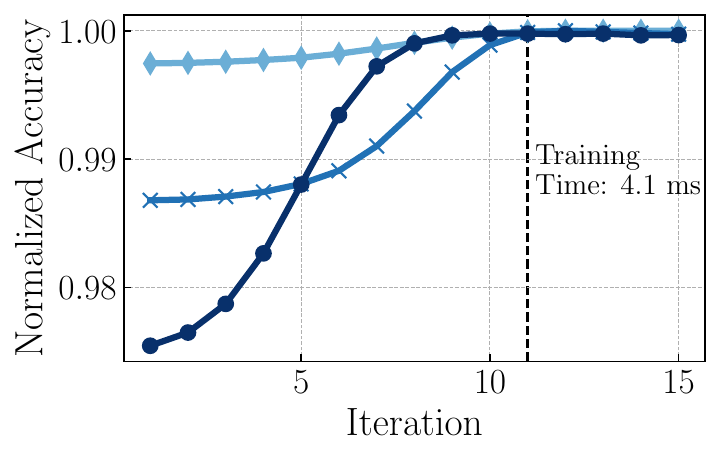}
        \caption{Parkinson's}
        \label{fig:online_parkinsons}
    \end{subfigure}
        \caption{Convergence of iterative re-training in \gls{TV} channels for different levels of the \gls{AR} factor $\eta$, considering $N_r=1024$ hidden nodes, learning rate $\gamma=0.5$, and batch size $|\mathcal{S}|=32$.}
        \label{fig:online-retraining}
\end{figure*}

In this section, we numerically evaluate the performance of \gls{XL}-\gls{MIMO}-\gls{ELM} on small-to-medium binary classification data sets.
Tests were performed on: \textit{i}) the Parkinson's data set~\cite{Parkinsons_dataset}; \textit{ii}) the Wisconsin Breast Cancer Diagnostic (WBCD) data set~\cite{breast_cancer_dataset}; \textit{iii}) a variation of MNIST~\cite{MNIST} where $100$ pixels have been randomly sampled as features and the problem was converted as binary classification by predicting whether each image contains an even or odd digit; and \textit{iv}) the SEmiCOnductor Manufacturing (SECOM) quality data set with $20$ randomly selected features. An $70\!:\!20\!:\!10$ training-test-validation split was applied, and all features have been independently standardized to zero mean and unit variance.
The considered benchmarks include: \textit{i}) a conventional \gls{ELM} algorithm (referred to as ``Digital ELM'') using sigmoid activation and uniformly random hidden layer coefficients in $[0,\!1]$; \textit{ii}) a \gls{MLP} of a hidden layer whose number of hidden neurons $n_h \triangleq N_r / (d+1)$ is chosen so that it has the same number of trainable weights as the \glspl{ELM}; and \textit{iii}) \glspl{SVM} which have $O(d)$ parameters.
All benchmarks constitute cases of typical classification approaches, i.e., they do not rely on \gls{OTA} computation nor on fading channels, therefore, if they were to be applied for \gls{EI}, one would have to account for the computational overhead needed for data transmission.

The mean test set classification accuracy over $300$ random seeds for the considered methods for all data sets is illustrated in Fig.~\ref{fig:all_datasets_wide} over increasing $N_r$ of \gls{RX} antennas (i.e., \gls{ELM}/\gls{MLP} trainable parameters). As expected from the theoretical analysis, as the number of hidden neurons increases, the \glspl{ELM} perform successful inference, while also approaching the performance of the \gls{MLP}, notwithstanding overfitting seen in Fig.~\ref{fig:secom}. It is also shown that, across all data sets, the performance of \gls{XL}-\gls{MIMO}-\gls{ELM} is on par with the digital \gls{ELM} variation under the same number of parameters, and it outperforms \gls{SVM} in most cases. Training results in a latency of a few to $O(100)$ ms for most cases, which typically fits within typical FR1 channel coherence times of low mobility~\cite{kivinen2001_emperical}.

Furthermore, to investigate the performance brought to the proposed \gls{XL}-\gls{MIMO}-\gls{ELM} framework directly by the considered nonlinear activation function, we have performed an ablation study by repeating the experiments of Fig.~\ref{fig:all_datasets_wide} while substituting the Rapp activation of~\eqref{eq:rapp-activation} with the standard Sigmoid and ReLU activations.
These activations are used only as part of this ablation study and are not intended to be realized in hardware, in contrast to the Rapp activation.
The obtained performance is reported in Fig.~\ref{fig::ablation-study}, where it can be observed that, apart from overfitting scenarios in the digital activations, all tested settings achieve comparable performance.
This is an attractive benefit for the proposed hardware-implemented Rapp activation. Its relevant robustness to overfitting is attributed to its parameterized form, since appropriate $\alpha$ and $y_{\rm sat}$ values can be determined through hyper-parameter tuning.
To further demonstrate the effects of hyper-parameter tuning, a sensitivity analysis is summarized in Fig.~\ref{fig:rapp-parameters}, where the combined effects of different values for $\alpha$ and $y_{\rm sat}$ are given for each data set. For the ranges of values investigated therein, performance variations are observed that range between $0.009$ and $0.025$ depending on the dataset, demonstrating relative robustness.

We have conducted further investigations on the effect of wireless parameters on the XL-MIMO-ELM performance. A \gls{GO} benchmark trained to account for varying fading, dubbed ``\gls{OTA}-\gls{DNN},'' has been implemented. In particular, the approach of~\cite{DeepOAC} has been adopted with a $10$-layer \gls{MLP} split across the transceivers. The network received both the data instance and the current channel matrix $\mathbf{H}^{\rm r}(k)$ as inputs. The outputs of its $5$-th hidden layer were transmitted over the channel which were then received as inputs by its $6$-th layer. The network was trained on joint batches of randomly sampled $\mathbf{H}^{\rm r}(k)$ paired with random $\mathbf{x}^{(i)} \in \mathcal{D}$ for $300$ epochs, resulting in excessive training periods. Figure~\ref{fig:acc-vs-snr} shows the effect of the \gls{SNR} level on the classification accuracy. I.i.d. \gls{AWGN} instances were sampled for every computation of~\eqref{eq:received signal} with $\sigma^2$ values so that the receive \gls{SNR} results in pre-defined levels.
It can be observed that, for SNR $\geq 20$ dB, the performance of the proposed scheme under all system parameters (data set and $N_r$) approaches the performances of its idealized case counterpart (Fig.~\ref{fig:all_datasets_wide}), demonstrating the validity of the high-\gls{SNR} analysis conducted.
Next, Fig.~\ref{fig:acc-vs-ricean} examines the effects of channel conditions on XL-MIMO-ELM. As discussed in Prop.~\ref{prop:universal-approximation}, rich scattering conditions provide favorable distributions with enough diversity for accurate classification, while results degrade as the dominance of the \gls{LOS} component increases, as shown in Fig.~\ref{fig:ricean_breast_cancer}. Correlated Rayleigh fading~\cite{6184250} is considered in Fig.~\ref{fig:spatial_correl}, where $\kappa=0$ and $\mathbf{H}_{\rm NLoS} \sim \mathcal{CN}(\mathbf{0}, \mathbf{R}_\rho \otimes \mathbf{R}_\rho)$ with $[\mathbf{R}_\rho]_{i,j} = \rho^{|i-j|}$, where $0 \leq \rho \leq 1$ is the correlation coefficient.
\gls{OTA}-\gls{DNN} presents an upper bound due to its larger model capacity and offline training overhead.
However, its performance is approached by our XL-MIMO-ELM for $\kappa \!\leq\! 10$ dB and $\rho \!< \!0.8$.

To evaluate the re-training procedure of~\eqref{eq:w-online-update} for \gls{TV} channels, the convergence of the weight update policy is given in Fig.~\ref{fig:online-retraining} for various \gls{AR} factors. The accuracy is normalized with respect to the one achievable by \eqref{eq:training}'s \gls{LS} solution over the complete data. It can be inferred that, as the channel moves to its next state, very few updates (each executed within approx. $ 0.4$~ms) are needed to re-tune the XL-\gls{MIMO}-\gls{ELM} to achieve its previous performance, even in low \gls{AR} regimes. It is finally evident that high AR factors enable faster convergence.

\section{Conclusion}
This paper introduced a novel framework for \gls{OTA} training and inference leveraging wireless propagation and nonlinear analog combining in XL MIMO systems. An analog activation function modeling soft thresholding was presented together with a theoretical analysis confirming that the proposed XL-MIMO-ELM implementation preserves the universal approximation capabilities of original \glspl{ELM}. A procedure for fast re-tuning under dynamic fading was also included. Our numerical investigations demonstrated that XL-\gls{MIMO}-\gls{ELM} achieves performance on par with the original \gls{ELM} algorithm and \gls{DNN} benchmarks across various setups.

\bibliographystyle{IEEEtran}
\bibliography{references}

\end{document}